\documentclass[aps,pra,superscriptaddress]{revtex4-2}
\usepackage[utf8]{inputenc}
\usepackage[colorlinks]{hyperref}
\date{\today}

\usepackage{tikz}
\usetikzlibrary{shapes.geometric, arrows}

\tikzstyle{startstop} = [rectangle, rounded corners, 
minimum width=3cm, 
minimum height=1cm,
text centered, 
draw=black, 
fill=red!30]

\tikzstyle{io} = [trapezium, 
trapezium stretches=true,
trapezium left angle=70, 
trapezium right angle=110, 
minimum width=3cm, 
minimum height=1cm, text centered, 
draw=black, fill=blue!30]

\tikzstyle{process} = [rectangle, 
minimum width=3cm, 
minimum height=1cm, 
text centered, 
text width=3cm, 
draw=black, 
fill=orange!30]

\tikzstyle{decision} = [diamond, 
minimum width=3cm, 
minimum height=1cm, 
text centered, 
draw=black, 
fill=green!30]
\tikzstyle{arrow} = [thick,->,>=stealth]

\usepackage{bm}
\usepackage{amssymb}
\usepackage{amsmath}
\usepackage{amsthm}
\usepackage{mathrsfs}
\usepackage{mathtools}
\DeclarePairedDelimiter{\bra}{\langle}{\rvert}%
\DeclarePairedDelimiter{\ket}{\lvert}{\rangle}%
\DeclarePairedDelimiter{\expval}{\langle}{\rangle}%
\DeclarePairedDelimiter{\abs}{\lvert}{\rvert}%
\DeclarePairedDelimiterX\innerp[2]{\langle}{\rangle}{#1\delimsize\vert\mathopen{}#2}%
\DeclarePairedDelimiterX\expvalOP[3]{\langle}{\rangle}{#1\,\delimsize\vert\,\mathopen{}#2\,\delimsize\vert\,\mathopen{}#3}%
\DeclarePairedDelimiterX\ketbra[2]{\lvert}{\rvert}{#1\delimsize\rangle\!\delimsize\langle#2}%
\DeclarePairedDelimiterX\outerp[2]{\lvert}{\rvert}{#1\delimsize\rangle\!\delimsize\langle#2}%
\DeclarePairedDelimiterX\projector[1]{\lvert}{\rvert}{#1\delimsize\rangle\!\delimsize\langle#1}%
\DeclareMathOperator{\tr}{tr}%
\DeclareMathOperator*{\spn}{span}

\newtheorem{theorem}{Theorem}
\newtheorem{corollary}{Corollary}
\newtheorem{lemma}{Lemma}
\newtheorem{definition}{Definition}

\begin{document}
\title{Local unitary decomposition of tripartite arbitrary leveled qudit stabilizer states into \(p\)-level-qudit EPR and GHZ states}
\author{Yat Wong}
\affiliation{Pritzker School of Molecular Engineering, University of Chicago, Chicago, Illinois 60637, USA}
\author{Liang Jiang}
\affiliation{Pritzker School of Molecular Engineering, University of Chicago, Chicago, Illinois 60637, USA}
\begin{abstract}
We study the entanglement structure of tripartite stabilizer states on \(N\) qudits of dimension \(D\), distributed across parties \(A\), \(B\), and \(C\), under arbitrary local unitaries. Prior work by Bravyi et al. and Looi et al. showed that qubit and squarefree qudit stabilizer states can be transformed via local Clifford unitaries into tensor products of GHZ states, EPR pairs, and unentangled qudits\cite{Bravyi_2006,Looi_2011}. We generalize this to arbitrary integer \(D\) by introducing local unitaries beyond the Clifford group, enabling decomposition of prime-power qudit stabilizer states into \(p\)-level GHZ states, EPR pairs, and unentangled qudits. Our algorithm leverages subsystem phase matrices to characterize entanglement and applies to quantum protocols requiring efficient entanglement distribution.
\end{abstract}

\maketitle
\section{Introduction}
Various quantum protocols have been proposed based on distributed quantum architectures. Quantum protocols, such as quantum key distribution, distributed quantum sensing, and distributed quantum computing, rely on efficient entanglement distribution across multiple nodes. Understanding and manipulating multipartite entangled states is thus critical for shaping entanglement into desired forms, such as GHZ or EPR states, for these applications. While general tripartite pure states are equivalent to bipartite mixed states and lack a full characterization beyond density matrices, specific classes like stabilizer states admit simpler descriptions. Bravyi et al. \cite{Bravyi_2006} showed that tripartite qubit stabilizer states are reducible via local Clifford unitaries to GHZ states, EPR pairs, and unentangled qubits, a result extended by Looi et al. \cite{Looi_2011} to squarefree qudit dimensions. Unlike prior work limited to squarefree \(D\), analyzing \(D=p^n\) involves handling states requiring more than \(N\) generators, where \(N\) is the number of qudits.

Using a novel tool, subsystem phase matrix, to capture entanglement properties, we extend these results to arbitrary qudit dimensions \(D\), including prime-power cases, by allowing local unitaries beyond the generalized Clifford group. Section \ref{sec:prelim} introduces notation and preliminaries, Section \ref{sec:props} discusses special properties of prime-power stabilizer states, Section \ref{sec:spm} defines the subsystem phase matrix, Section \ref{sec:lemmas} presents key lemmas, and Section \ref{sec:main} details our main decomposition theorems.
\section{Preliminaries and Notation}\label{sec:prelim}
We consider a system of \(N\) qudits, each with dimension \(D\), partitioned into subsystems \(A\), \(B\), and \(C\).
The single qudit generalized Pauli operators, also known as Weyl–Heisenberg operators, are generated by the shift operator \(X=\sum_j\ket{j+1}_i\bra{j}_i\) and the clock operator \(Z=\sum_j\omega^j\ket{j}\bra{j}\). They can be represented by a two-dimensional free module element \(\left(x,z\right)\) and an extra number \(\gamma\) to account for the phase, i.e. \(\omega^\gamma\sigma_{x,z}=\omega^{\gamma} X^xZ^z\), where \(\omega=e^{\frac{2\pi i}{D}}\).  \(\gamma\) is always an integer if \(D\) is odd, and either an integer or half-integer for even \(D\). \(N-\)qudit Pauli operator \(\sigma\in\mathbb{P}_D^N\) can thus be represented by \(\omega^\gamma\sigma_\mu=\omega^{\gamma}\prod_iX_i^{x_{\mu,i}}Z_i^{z_{\mu,i}}\), where \(\mu=(\vec{x}_\mu^T,\vec{z}_\mu^T)^T\) is the \(2N\)-dimensional module element representing how many times shift/clock operators are applied on each qudit. Note that modules defined over \(\mathbb{Z}_D\) is a vector space if and only if \(D\) is a prime, since \(\mathbb{Z}_D\) would then be a field. For squarefree but composite \(D\), the module can be written as direct products of vector spaces due to the Chinese remainder theorem, hence most properties of stabilizer states with prime \(D\) carries over to stabilizer states with squarefree composite \(D\). 
The multiplication rule and thus commutation rule between \(\sigma\)-operators are given by \begin{equation}\sigma_\mu\sigma_\nu=\omega^{\vec{z}_\mu\cdot\vec{x}_\nu}\sigma_{\mu+\nu}=\omega^{\nu^T\Omega_N \mu}\sigma_\nu\sigma_\mu,\end{equation}
where \(\Omega_N=\begin{pmatrix}0&I_N\\-I_N&0\end{pmatrix}\).
A generalized Pauli operator \(\sigma\) is a stabilizer of a state \(\ket{\psi}\) if and only if the state is an eigenstate of the operator with eigenvalue 1, i.e. \(\sigma\ket{\psi}=\ket{\psi}\), thus not all Pauli operators are a stabilizer of some state, such as \(\omega I\).
A generalized Clifford unitary \(U\in\mathbb{C}\) transforms the set of generalized Pauli operators into itself. Due to the properties of the generalized Pauli group, especially the commutation rules, we can show that for any \(U\) there exists \(M\) such that
\begin{equation}U\sigma_{\mu}U^\dagger\propto\sigma_{M\mu}\end{equation}
where \(M\) is a symplectic matrix that satisfies \(M\Omega_N M^T=\Omega_N\).
Generalized Clifford group, the group of all generalized Clifford unitaries, can be represented by all possible sympectic \(M\), and is generated by phase gate \(P=\sum_j\omega^{j^2/2}\ket{j}\bra{j}\), qudit Fourier transform (also generalized Hadamard gate) \(H=D^{-1/2}\sum_{jk}\omega^{jk}\ket{j}\bra{k}\), and a two-qudit gate Controlled-Z \(CZ=\sum_{jk}\omega^{jk}\ket{j,k}\bra{j,k}\), up to a generalized Pauli operator and a global phase. From here, we denote generalized Pauli and generalized Clifford as Pauli and Clifford.
A stabilizer group, an Abelian subgroup of Pauli operators where the only element proportional to the identity is identity itself, can be represented as an isotropic submodule where the only element with \(\vec{x}=\vec{z}=\vec{0}\) must satisfy \(\gamma=0\).
When a stablizer group \(S\) has the same number of elements as the dimension of the Hilbert space, i.e. \(\abs{S}=D^N\), the stabilizer group stabilizes only one pure state, i.e.
\begin{equation}\forall s\in S, s\ket{\psi}=\ket{\psi} \Rightarrow \projector{\psi}=D^{-N}\sum_{s\in S}s.\end{equation}
Hence, we define the stabilizer state specified by \(S\) as \(\projector{S}=D^{-N}\sum_{s\in S}s\).
For composite but not prime-power \(D\), Theorem 4 of \cite{Looi_2011} establishes an isomorphism to a tensor product of prime-power \(D'_i=p_i^{n_i}\) stabilizer formalisms, where \(\prod_iD'_i=D\), allowing us to focus on \(D=p^n\).
\section{Special Properties of Prime Power Dimension Stabilizer States}\label{sec:props}

Unlike squarefree \(D\), where stabilizer groups have a dimension up to \(N\), prime-power \(D=p^n\) stabilizer groups may have dimension up to \(2N\), complicating their entanglement structure. As a result, the dimension of a stabilizer group representing a single pure state ranges from \(N\) to \(2N\), whereas squarefree \(D\) stabilizer states can always be represented with a stabilizer group with dimension \(N\). It can be shown that such dimension does not vary under \(D\)-level Clifford operations, and hence there are pairs of stabilizer states that are not connected to each other through Clifford operations. For instance, for \(D=4\), the state \(\frac{1}{\sqrt{2}}\left(\ket{0}+\ket{2}\right)\) is stabilized by \(\expval{X^2,Z^2}\), and there is no Clifford unitary that prepares this state from \(\ket{0}\). One implication of this is not all stabilizer states can be represented with graph states, since the stabilizer group of graph states must have \(N\) generators. Note that for composite squarefree \(D\), there exists states that look similar, yet they are still connected to \(\ket{0}^{\otimes N}\) with Clifford unitaries.

Such phenomena is not limited to a single qudit. One remarkable example is the \(D=9\) \(N=3\) state stabilized by \(S=\expval{X_1X_2X_3,X_1^3X_2^6,Z_1Z_2Z_3,Z_1^3Z_2^6}\) is not connected to the GHZ state, stabilized by \(S'=\expval{X_1X_2X_3,Z_1Z_2^8,Z_1Z_3^8}\), through any Clifford operation, yet they are connected by some qudit-local unitaries. The two states can be written as products of two \(D=3\) GHZ states in different ways:
\begin{equation}
\ket{S}=\frac{1}{3}\left(\ket{000}+\ket{012}+\ket{021}+\ket{102}+\ket{111}+\ket{120}+\ket{201}+\ket{210}+\ket{222}\right)\otimes\frac{1}{\sqrt{3}}\left(\ket{000}+\ket{111}+\ket{222}\right),
\end{equation}
\begin{equation}
\ket{S}'=\frac{1}{\sqrt{3}}\left(\ket{000}+\ket{111}+\ket{222}\right)\otimes\frac{1}{\sqrt{3}}\left(\ket{000}+\ket{111}+\ket{222}\right).
\end{equation}
However, due to the structure of the prime-power-dimensional qudit Pauli group, one cannot rotate the first part without affecting the second part with \(D=9\) Clifford unitaries, since Clifford unitaries cannot alter the dimension of a subgroup of the Pauli group. As demonstrated by this state, analysis based only on Clifford operations cannot thoroughly analyze the entanglement properties of stabilizer states with prime power \(D\). To accommodate such states, we introduce a tool that is invariant under arbitrary local Clifford operations and simultaneously allow some local unitaries beyond Clifford group.
\section{Definition of Subsystem Phase Matrix}\label{sec:spm}
Most entanglement characterization of stabilizer states with squarefree \(D\) relies on having the same number of generators as the number of qudits, which is generally invalid for prime power \(D\). To characterize local Clifford invariants in such situation, we define subsystem phase matrices, which capture the commutation relations of stabilizers across parties. 
\begin{definition}[Subsystem Phase Matrix]
A set of \(m\) \(\left(\mathbb{Z}_{p^n}\right)^{\mathcal{N}\times\mathcal{N}}\) matrices \(\left\{M_\alpha\right\}\) is a valid set of subsystem phase matrices for an \(m\)-partite \(D=p^n\) stabilizer state \(\ket{S}\) if and only if there exists a pair of functions \(f:S\rightarrow\left(\mathbb{Z}_{p^n}\right)^\mathcal{N}\) and \(F:\left(\mathbb{Z}_{p^n}\right)^\mathcal{N}\rightarrow S\) such that
\begin{equation}
\forall s,s'\in S,s_\alpha s_\alpha'=\omega^{f^T(s')M_\alpha f(s)}s_\alpha's_\alpha,
\end{equation}
\begin{equation}
\forall v,v'\in \left(\mathbb{Z}_{p^n}\right)^\mathcal{N},F(v)_\alpha F(v')_\alpha=\omega^{v'^TM_\alpha v}F(v')_\alpha F(v)_\alpha,
\end{equation}
where \(\mathcal{N}\) is an integer and \(\{\alpha\}=\{1,2,\dots,m\}\) is the set of subsystems. Note that for any product operator \(\hat{O}\), we represent its part in subsystem \(\alpha\) as \(\hat{O}_\alpha\), such that \(\prod_\alpha \hat{O}_\alpha=\hat{O}\). Equivalently, given a generating set of stabilizers \(\left\{g^{(i)}\right\}\), a valid set of subsystem phase matrices \(\left\{M_{\alpha,ij}\right\}\) can be obtained through the module representation of the generators:
\begin{equation}M_{\alpha,ij}=r^{(i)^T}\Omega_\alpha r^{(j)}\end{equation}
where \(r^{(i)}\) is the module element corresponding to the stabilizer \(g^{(i)}\), i.e. \(\sigma_{r^{(i)}}\propto g^{(i)}\), and \(\Omega_\alpha\) is the bilinear form restricted to only subsystem \(\alpha\). In terms of operators, this corresponds to
\begin{equation}g^{(j)}_{\alpha}g^{(i)}_{\alpha}=\omega^{M_{\alpha,ij}}g^{(i)}_{\alpha}g^{(j)}_{\alpha}.\end{equation}
\end{definition}
Note that these matrices are antisymmetric because \(\Omega_\alpha=-\Omega_\alpha^T\), and the sum of phase matrices over all subsystems gives 0, since all stabilizers commute, i.e.
\begin{equation}
M_{\alpha,ii}=0,M_{\alpha}+M_{\alpha}^T=0,\sum_\alpha M_\alpha=0. \label{eq:SPM}
\end{equation} 
It is trivial to see that such matrices are invariant under local Clifford unitaries, and any two stabilizer states that are equivalent up to some local Clifford unitaries must have the same matrices up to some basis transform. These subsystem phase matrices contain all entanglement information of tripartite stabilizer states, which will be proven in a subsequent section. In the simplified case of a pure bipartite stabilizer state, diagonalizing such matrices identifies pairs of generators that are non-commuting in subsystems, revealing the extractable entangled qudit pairs. 

Such matrices are also invariant under some local non-Clifford unitaries. For example, if a stabilizer state \(\ket{\psi}\) has \(Z_1^{p^{n-1}}\) as a stabilizer, the following non-Clifford unitary
\begin{equation}
V_1=\frac{1}{\sqrt{p}}\sum_{j=0}^{p^{n-1}-1}\sum_{k=0}^{p-1}\sum_{l=0}^{p-1}\omega^{p^{n-1}kl}\ket{p^{n-1}k+j}\bra{pj+l}_1
\end{equation}
can transform the state into another stabilizer state without altering the subsystem phase matrix. If we decompose the \(D=p^n\) qudit into \(n\) \(D=p\) qudits with \(p-\)nary representation, this unitary corresponds to applying Hadamard gate on the lowest qudit and then moving it to the highest. We can see that \(V_1^\dagger Z_1^pV_1\ket{\psi}=Z_1\ket{\psi}\) and \(V_1^\dagger X_1V_1\ket{\psi}=X_1^p\ket{\psi}\), thus \(V_1\) scales up \(Z_1\) and scales down \(X_1\), resulting in a new stabilizer state with \(X_1^{p^{n-1}}\) as an extra stabilizer generator. The subsystem phase matrix is unaltered because the scaling cancels out.
In the following sections, we also use the projection of the subsystem phase matrix into the subring \(\mathbb{Z}_p\), \(M'_a\), which is defined as
\begin{equation}M'_{a,ij}=M_{a,ij}\bmod{p}.\end{equation}
\section{Important Lemmas}\label{sec:lemmas}
\begin{lemma}\label{commute}
Given a stabilizer group \(S\) that specifies a unique \(D\)-level qudit pure stabilizer state, any Pauli operator \(g\) that commutes with all stabilizers must be a stabilizer up to some phase, i.e.
\begin{equation}\abs{S}=D^N,\forall s\in S,\left[g,s\right]=0\Rightarrow\exists\phi\text{ s.t. }\omega^\phi g\in S,\end{equation}
where \(\phi\) is an integer for odd \(D\) and integer/half-integer for even \(D\).
\end{lemma}
\begin{proof}
\begin{equation}
g\projector{S}g^\dagger=g\left(D^{-N}\sum_{s\in S}s\right)g^\dagger=D^{-N}\sum_{s\in S}gsg^\dagger=D^{-N}\sum_{s\in S}s=\projector{S}.
\end{equation}
Therefore, \(g\ket{S}\propto\ket{S}\) and \(\abs{\expvalOP{S}{g}{S}}=1\). To show that \(\phi\) must be an integer/half-integer, consider
\begin{equation}
\expvalOP{S}{g}{S}=\tr\left(g\projector{S}\right)=D^{-N}\sum_{s\in S}\tr{gs^\dagger}.
\end{equation}
Since both \(g\) and \(s\) are elements of the Pauli group, \(\tr{gs^\dagger}=0\) if \(g\not\propto s\) and \(\tr{gs^\dagger}=D^N\omega^{\phi}\) if \(g=\omega^\phi s\). Because there must not be two elements that differ only by a phase in a stabilizer group, there must be one and only one \(s\) that is proportional to \(g\), and \(gs^\dagger\) is proportional to identity. Since \(\phi\) is an integer/half-integer if and only if \(\omega^\phi I\) is a Pauli, it must be true.
\end{proof}
\begin{lemma}\label{extract}
If a \(D=p^n\) pure stabilizer state is stabilized by \(\prod_{r\in R}X_r\) and pairwise \(Z_r^{p^{n-n'}}Z_{r'}^{-p^{n-n'}}\) over some number of qudits \(R\), where \(\abs{R}\geq2\) and \(1\leq n'\leq n\), then there is some local unitary operation that extracts a \(D'=p^{n'}\) \(\abs{R}\)-qudit GHZ/EPR state.
\begin{equation}
X_r^{\otimes_{r\in R}}\in S,\forall r,r'\in R,Z_r^{p^{n-n'}}Z_{r'}^{-p^{n-n'}}\in S\Rightarrow\exists\left\{U_r\right\}\text{ s.t. }\left(\bigotimes_{r\in R}U_r\right)\left(\ket{S}\otimes\ket{0}_{p}^{\otimes r\in R}\right)=\ket{\tilde{S}}\otimes\ket{GHZ}_{p^{n'}},
\end{equation}
where
\begin{equation}\tilde{S}=\bar{S}\otimes\left<Z_{r_1}^{p^{n-n'}}\right>,\end{equation}
\begin{equation}\bar{S}=\left\{s\in S\middle|\left[s,Z_r^{p^{n-n'}}\right]=0\right\},\end{equation}
\begin{equation}\ket{GHZ}_{p^{n'}}=p^{-n'/2}\sum_{k=0}^{p^{n'}-1}\ket{k}_r^{\otimes r\in R}.\end{equation}
Note that if \(\left|R\right|=2\), the extracted state is an EPR pair, otherwise it is a GHZ state.
\end{lemma}
\begin{proof}
Write the state as a summation over the lowest \(n'\) sub-\(p\)-it of the \(\abs{R}\) involved qudits in the \(p\)-nary representation
\begin{equation}\ket{S}=\sum_{\vec{v}}w_{\vec{v}}\ket{\varphi_{\vec{v}}}\otimes\ket{\vec{v}}.\end{equation}
For the state to be stabilized by the pairwise \(Z_r^{p^{n-n'}}Z_{r'}^{-p^{n-n'}}\), only terms with \(\vec{v}\propto\vec{1}\) can be nonzero, where \(\vec{1}=(1,1,\dots,1)^T\) is the \(\abs{R}\)-dimensional vector/module element of all \(1\)s, i.e.
\begin{equation}\ket{S}=\sum_{j=0}^{p^{n'}-1}w_{j\vec{1}}\ket{\varphi_{j\vec{1}}}\otimes\ket{j\vec{1}},\end{equation}
where \(j\vec{1}=(j,j,\dots,j)^T\). \(\prod_{r\in R}X_r\) further limits all weights and state with different \(j\) to be same, i.e.
\begin{equation}\ket{S}=\sum_{j=0}^{p^{n'}-1}w_{\vec{0}}\ket{\varphi_{\vec{0}}}\otimes\ket{j\vec{1}},\end{equation}
where \(\vec{0}=(0,0,\dots,0)^T\). Hence,
\begin{equation}\ket{S}=\ket{\varphi_{\vec{0}}}\otimes\ket{GHZ}_p.\end{equation}
It is clear that every stabilizer can be split into two parts: one operating on \(\ket{\varphi}\) and one operating on the GHZ/EPR state. In other words,
\begin{equation}
S=\bar{S}\otimes\left\{\prod_{r\in R}X_r^j\mathrel{\bigg|}0\leq j\leq p^{n'}-1\right\}.
\end{equation}
Note that \(\bar{S}\) includes \(\prod_{r\in R}X_r^{p^{n'}}\) and \(Z_r^{p^{n-n'}}Z_{r'}^{-p^{n-n'}}\). Swapping the GHZ/EPR state with \(\ket{0}_{p^{n'}}^{\otimes r\in R}\) gives \(\ket{\varphi_{\vec{0}}}\otimes\ket{0}_{p^{n'}}^{\otimes r\in R}\), and it is trivial to show that this state is stabilized by \(\tilde{S}\).
\end{proof}
Note that we added \(\ket{0}_p^{n'}\) states to effectively store GHZ/EPR states, but in principle the procedure can be done in place by performing analysis of Pauli groups of a mixture of \(D=p,p^2,\dots,p^{n}\) qudits using symplectic modules~\cite{symplectic}, and update the definition of Pauli group and Clifford group between every step.
\section{Main results}\label{sec:main}
We present the main corollary, stating that all \(D=p^n\) tripartite pure stabilizer states can be reduced to \(D=p\) GHZ states, EPR pairs, and unentangled qudits with local unitaries, followed by the two theorems that enables the corollary. These theorems hold generally for \(m\)-partite stabilizer states, with \(m=3\) corresponding to the tripartite case.
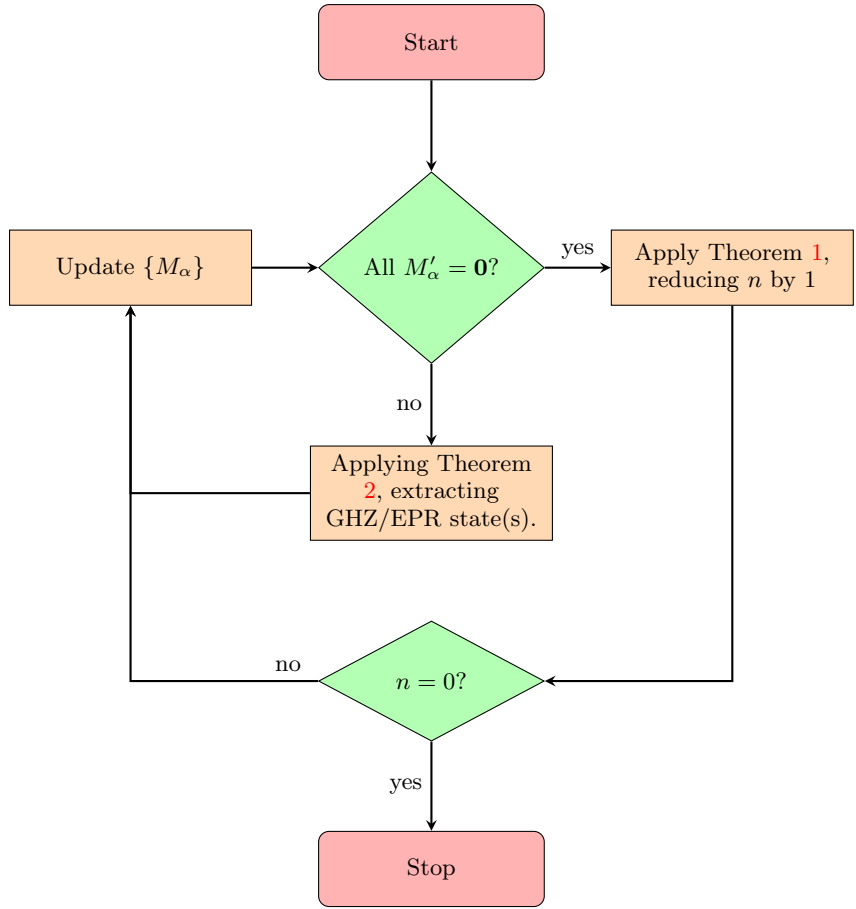
\begin{figure}
    \centering

    \begin{tikzpicture}[node distance=2cm]
    
    \node (start) [startstop] {Start};
    \node (checkm) [decision, below of=start, yshift=-1cm] {All \(M'_\alpha=\bold{0}\)?};
    \node (updatem) [process, left of=checkm, xshift=-2cm] {Update \(\left\{M_\alpha\right\}\)};
    \node (extract) [process, below of=checkm, yshift=-1cm] {Applying Theorem \ref{thm:2}, extracting GHZ/EPR state(s).};
    \node (reduce) [process, right of=checkm, xshift=2cm] {Apply Theorem \ref{thm:1}, reducing \(n\) by 1};
    \node (checkn) [decision, below of=extract, yshift=-0.5cm] {\(n=0\)?};
    \node (stop) [startstop, below of=checkn, yshift=-0.5cm] {Stop};
    
    \draw [arrow] (start) -- (checkm);
    \draw [arrow] (updatem) -- (checkm);
    \draw [arrow] (checkm) -- node[anchor=east] {no} (extract);
    \draw [arrow] (checkm) -- node[anchor=south] {yes} (reduce);
    \draw [arrow] (checkn) -| node[anchor=west, yshift=0.2cm, xshift=1.8cm] {no} (updatem);
    \draw [arrow] (checkn) -- node[anchor=east] {yes} (stop);
    \draw [arrow] (reduce) |- (checkn);
    \draw [arrow] (extract) -| (updatem);
    
    \end{tikzpicture}
    \caption{Flowchart of entanglement extraction}
    \label{fig:flowchart}
\end{figure}
\begin{corollary}
There exist some local unitaries that convert a tripartite pure \(D=p^n\) stabilizer state into a tensor product of GHZ states, EPR pairs, and unentagled qudits over \(p\)-level systems.
\begin{multline*}
\exists U_a\in \mathbb{U}_a,U_b\in \mathbb{U}_b,U_c\in \mathbb{U}_c\text{ s.t. }U_a\otimes U_b\otimes U_c\ket{S}\\=\ket{GHZ}_p^{\otimes N_{GHZ}}\otimes\ket{EPR}_{p,ab}^{\otimes N_{ab}}\otimes\ket{EPR}_{p,ac}^{\otimes N_{ac}}\otimes\ket{EPR}_{p,bc}^{\otimes N_{bc}}\otimes\ket{0}_{p,a}^{\otimes N_a}\otimes\ket{0}_{p,b}^{\otimes N_b}\otimes\ket{0}_{p,c}^{\otimes N_c}.
\end{multline*}
\end{corollary}
\begin{proof}
The subsystem phase matrices of any tripartite pure \(D=p^n\) stabilizer state must satisfy at least one of the following conditions:
\begin{enumerate}
  \item All projected subsystem phase matrices are trivial.
  \begin{equation}
\forall\alpha,M'_\alpha=\bold{0}
  \end{equation}
  \item The intersection of the spans of all subsystem phase matrices is nontrivial.
  \begin{equation}
\bigcap_\alpha\spn(M_\alpha)\neq\{\bold{0}\}
  \end{equation}
  \item There exists at least one \(\vec{v}\) such that \(\vec{v}\) is in the span of two subsystem phase matrices and \(p^{n-1}\vec{v}\) is not in the span of the remaining subsystem phase matrix.
  \begin{equation}
\exists\alpha,\vec{v}\text{ s.t. }\vec{v}\in\bigcap_{\beta\neq\alpha}\spn(M_\beta),p^{n-1}\vec{v}\notin\spn(M_\alpha)
  \end{equation}
\end{enumerate}
These conditions correspond to Theorem \ref{thm:1}, Corollary \ref{cor:2}, and Corollary \ref{cor:3} respectively. Theorem \ref{thm:1} allows reduction of \(n\) by \(1\), while Corollary \ref{cor:2} and \ref{cor:3}, implications of Theorem \ref{thm:2}, allow extraction of entanglement in the form of \(D=p\) GHZ/EPR state(s). We can then follow the entanglement extraction flowchart shown in Fig. \ref{fig:flowchart} to iteratively reduce \(n\) while extracting GHZ/EPR state(s).

To show that at least one of the conditions must hold, consider the situation where conditions 1 and 2 are false, i.e. at least one projected subsystem phase matrix is nontrivial, and the intersection of all subsystem phase matrices is trivial. Without loss of generality, assume \(M'_a\neq\bold{0}\), let \(\vec{\bar{v}}\) be a nonzero column of \(M'_a\), and \(\vec{v}\) be the corresponding column of \(M_a\). Since \(\vec{\bar{v}}\) is a nontrivial element, \(\vec{v}\) must be of order \(p^n\). Since \(p^{n-1}\vec{v}\in\spn(M_a)\), condition 2 being false implies that \(p^{n-1}\vec{v}\) must be absent from at least one of \(\spn \left(M_b\right)\) or \(\spn \left(M_c\right)\), without loss of generality assume \(p^{n-1}\vec{v}\notin\spn \left(M_c\right)\). Then, there exists invertible \(\mathcal{N}\times\mathcal{N}\) matrix \(L\) such that
\begin{subequations}
\begin{align}
L\vec{v}&=\begin{pmatrix}
1\\
\hline
0\\
\vec{0}
\end{pmatrix},\label{lv}\\
LM_cL^T&=
\left(\begin{array}{@{}c|c@{}}
0&\begin{matrix}0&\vec{0}^T\end{matrix} \\
\hline
\begin{matrix}0\\\vec{0}\end{matrix}&\widetilde{M}_c
\end{array}\right).\label{lcl}\\
LM_aL^T&=
\left(\begin{array}{@{}c|c@{}}
0&\begin{matrix}1&\vec{0}^T\end{matrix} \\
\hline
\begin{matrix}-1\\\vec{0}\end{matrix}&\widetilde{M}_a
\end{array}\right),\label{lal}\\
LM_bL^T&=
\left(\begin{array}{@{}c|c@{}}
0&\begin{matrix}-1&\vec{0}^T\end{matrix} \\
\hline
\begin{matrix}1\\\vec{0}\end{matrix}&-\widetilde{M}_a-\widetilde{M}_c
\end{array}\right).\label{lbl}
\end{align}
\end{subequations}
We justify the existence of \(L\) in the following four steps:
(1) For a vector \(\vec{v}\) of order \( p^n \), there exists a matrix \(L^{(1)}\) that satisfies Eq.~\ref{lv}.
(2) Since \( p^{n-1} \vec{v} \notin \spn(M_c) \), we can eliminate the entries in the first row and first column of \( M_c \) to obtain \(L^{(2)}\) that satisfies Eq.~\ref{lcl} by applying similarity transformations over rows and columns, while preserving the structure of Eq.~\ref{lv}.
(3) Because \( \vec{v} \in \spn(M_a) \), the first row and first column of \( M_a \) remain order \(p^n\) after the above transformation. By performing additional similarity transformations, we can obtain \(L^{(3)}\) that satisfies Eq.~\ref{lal}, while maintaining the structure of Eqs.~\ref{lv} and \ref{lcl}.
(4) Finally, we obtain Eq.~\ref{lbl} using Eq.~\ref{eq:SPM}.

According to Eqs.~\ref{lal} and \ref{lbl}, we have \(\vec{v}'=L^{-1}
\begin{pmatrix}
0\\
\hline
1\\
\vec{0}
\end{pmatrix}\) is in the span of both \(M_a\) and \(M_b\). Due to \(p^{n-1}\vec{v}'\in M_a\cap M_b\), condition 2 being false implies \(p^{n-1}\vec{v}'\notin M_c\), thus \(\vec{v}'\) must satisfy condition 3. Therefore, for any \(n\geq 1\), the subsystem phase matrices must satisfy at least one condition. By iteratively applying the theorems according to the flowchart in Fig. \ref{fig:flowchart}, any tripartite stabilizer state can be decomposed into \(p\)-level GHZ states, EPR pairs, and unentangled qudits. Since entanglement cannot be infinite for a state with finite qudits, eventually we will reach a product state, i.e. \(M_a=M_b=M_c=0\), with \(n=1\), which we can rotate all qudits to \(\ket{0}_p\). Ancillary qudits introduced in Theorem \ref{thm:2} are not required to extract the entanglement but simplifies the process. Note that although the unitaries involved are always Cliffords for some power of \(p\), the combination is generally not a \(D=p^n\) Clifford unitary.
\end{proof}
\begin{theorem}\label{thm:1}
If \(\forall \alpha,M'_\alpha=0\), then the \(m\)-partite stabilizer state is local-Clifford equivalent to a product of a \(D'=p^{n-1}\) stabilizer state and unentangled \(D''=p\) qudits.

\begin{equation}
\forall \alpha,M_\alpha\equiv0\left(\bmod{\;p}\right)\Rightarrow\exists U_\alpha\in\mathbb{C}_\alpha\text{ s.t. }\left(\bigotimes_\alpha U_\alpha\right)\ket{S}=\ket{\tilde{S}'}_{p^{n-1}}\otimes\ket{0}_p^{\otimes N}.
\end{equation}
\end{theorem}
\begin{proof}
Since \(M'_\alpha=0\), projecting the stabilizer group into \(D=p\) Pauli group results in a stabilizer group where every pair of stabilizer commutes in every subsystem, and hence there exists a set of local \(D=p\) Clifford unitaries that transform the projected stabilizer group to consist of only products of \(Z\)s. Hence, there exists some local \(D=p^n\) Clifford that transforms the original stabilizer group to consist of only products of \(Z\)s and \(X^p\)s. In the resultant stabilizer state, \(Z^{p^{n-1}}\) of every single qudit commutes with all stabilizers, and thus, from Lemma \ref{commute}, is a stabilizer up to some phase, which can be removed by applying \(X\)s. As the state is now stabilized by \(Z^{p^{n-1}}\) of individual qudits, if we write the computational basis in \(p\)-nary representation, the lowest \(p\)-dit is 0 for every qudit, i.e. 
\begin{equation}\bigotimes_\alpha U_\alpha\ket{S}=\ket{\tilde{S}}=\ket{\psi}_{p^{n-1}}\otimes\ket{0}^{\otimes N}.\end{equation}
To show that \(\ket{\psi}_{p^{n-1}}\) is a stabilizer state, consider
\begin{equation}X^p\ket{\tilde{S}}=\left(X'\ket{\psi}_{p^{n-1}}\right)\otimes\ket{0}^{\otimes N},\end{equation}
\begin{equation}Z\ket{\tilde{S}}=\left(Z'\ket{\psi}_{p^{n-1}}\right)\otimes\ket{0}^{\otimes N},\end{equation}
and since \(\tilde{S}\) consists of only products of \(Z\)s and \(X^p\)s, replacing \(X^p\) with \(X\) results in a \(D=p^{n-1}\) stabilizer group.
\end{proof}
\begin{theorem}\label{thm:2}
For a pure \(D=p^n\) stabilizer state, if there exists an order-\(p^n\) module element \(\vec{v}\) such that \(p^{n-n'}\vec{v}\) is in the span of \(m'\) subsystem phase matrices, where \(1\leq n'\leq n\), and \(p^{n-1}\vec{v}\) is not in the span of the sum of the \(m'\) subsystem phase matrices, then we can extract \(n'\) copies of \(m'\)-partite \(D=p\) GHZ/EPR state. Note that if \(m'=2\), the extracted state(s) is/are EPR pairs.
\begin{multline}
\exists\mathfrak{B}\subseteq\{\alpha\},\vec{v}\text{ s.t. }\vec{v}\not\equiv\vec{0}\left(\bmod{p}\right),p^{n-n'}\vec{v}\in\bigcap_{\beta\in\mathfrak{B}}\spn \left(M_\beta\right),p^{n-1}\vec{v}\notin\spn\left(\sum_\beta M_{\beta}\right)\\
\Rightarrow\exists U_\beta\in \mathbb{U}_\beta\text{ s.t. }\left(\bigotimes_{\beta\in\mathfrak{B}} U_\beta\right)\ket{S}\otimes\ket{0}_{p^{n'},\mathfrak{B}}=\ket{\tilde{S}}\otimes\ket{GHZ}_{p,\mathfrak{B}}^{\otimes n'}.
\end{multline}
\end{theorem}
\begin{proof}
Since \(p^{n-1}\vec{v}\notin\spn\left(\sum_{\beta\in\mathfrak{B}} M_{\beta}\right)\) and \(\sum_{\beta\in\mathfrak{B}} M_{\beta}\) are antisymmetric, there exists \(\vec{v}'\in\ker\left(\sum_{\beta\in\mathfrak{B}} M_{\beta}\right)\) such that \(\vec{v}'^T\vec{v}=1\). There also exists \(\vec{u}^{(\beta)}\) such that \(M_{\beta}\vec{u}^{(\beta)}=p^{n-n'}\vec{v}\). We now have
\begin{equation}\label{eq:thm2com}
\left(\sum_{\beta\in\mathfrak{B}} M_{\beta}\right)\vec{v}'=\vec{0},
\end{equation}
\begin{equation}
\vec{v}'^TM_\beta\vec{u}^{(\beta)}=p^{n-n'},
\end{equation}
\begin{equation}\label{eq:thm2com3}
p^{n'}M_\beta\vec{u}^{(\beta)}=0.
\end{equation}
Denote the stabilizers corresponding to \(\vec{v}'\) and \(\vec{u}^{(\beta)}\) as \(g\) and \(h^{(\beta)}\) respectively, and let \(k_\beta\) be the minimal integer satisfying \(h^{(\beta)^{p^{k_\beta}}}_{\beta}\propto I\). Then, there exists local Clifford that rotates \(h^{(\beta)}_\beta\) to
\begin{equation}
h'^{(\beta)}_\beta\propto Z_{\beta,1}^{p^{n-k_\beta}}.
\end{equation}
where \(Z_{\beta,1}\) indicates a clock operator on the first qubit of the subsystem \(\beta\). If \(k_\beta>n'\), then Eq. \ref{eq:thm2com3} indicates \(Z_{\beta,1}^{p^{n-k_\beta+n'}}\) commutes with all stabilizers, and hence is a stabilizer up to some phase. Apply suitable \(X_{\beta,1}\) to erase the phase, then we can apply \(V_{\beta,1}\) \(k_\beta-n'\) times, resulting in
\begin{equation}
h''^{(\beta)}_\beta\propto Z_{\beta,1}^{p^{n-n'}}.
\end{equation}
Now we have \(Z_{\beta,1}^{p^{n-n'}}g'=\omega^{p^{n-n'}} g'Z_{\beta,1}^{p^{n-n'}}\), so there exists local Clifford that conserves all \(h'^{(\beta)}_\beta\) and rotates \(g'\) to
\begin{equation}
g''=\bigotimes_{\beta\in\mathfrak{B}} X_{\beta,1}\otimes \bigotimes_{\alpha\notin\mathfrak{B}}g_\alpha.
\end{equation}
From Lemma \ref{commute}, Eq. \ref{eq:thm2com} indicates that now \(\bigotimes_{\beta\in\mathfrak{B}} X_{\beta,1}\) is a stabilizer up to some phase, while \(M_\beta\vec{u}^{(\beta)}-M_{\beta'}\vec{u}^{(\beta')}=0\) indicates that \(Z_{\beta,1}^{p^{n-n'}}Z_{\beta',1}^{-p^{n-n'}}\) is also a stabilizer up to some phase. After erasing the phases by applying suitable \(X\)s and \(Z\)s, we can extract an \(m'\)-partite \(D=p^{n'}\) GHZ/EPR state using Lemma \ref{extract}, which is equivalent to \(n'\) copies of \(D=p\) GHZ/EPR states.
\end{proof}
Based on Theorem \ref{thm:2}, we can take special values for \(n'\) and \(\mathfrak{B}\) to obtain the following two corollaries:
\begin{corollary}\label{cor:2}
For a triparite pure \(D=p^n\) stabilizer state, if the intersection of the spans of all subsystem phase matrices is nontrivial, then GHZ state(s) can be extracted.
\end{corollary}
\begin{proof}
Let \(\vec{\bar{v}}\) be a nontrivial element of \(\bigcap_\alpha\spn(M_\alpha)\) and its order be \(p^{n'}\), then there exists an order-\(p^n\) \(\vec{v}\) such that \(p^{n-n'}\vec{v}=\vec{\bar{v}}\). This satisfies the conditions of Theorem \ref{thm:2} with \(\mathfrak{B}=\{a,b,c\}\), hence \(n'\) copies of GHZ state(s) can be extracted.
\end{proof}
\begin{corollary}\label{cor:3}
For a triparite pure \(D=p^n\) stabilizer state, if there exists at least one \(\vec{v}\) such that \(\vec{v}\) is in the span of two subsystem phase matrices and \(p^{n-1}\vec{v}\) is not in the span of the remaining subsystem phase matrix, then \(n\) copies of EPR pair(s) can be extracted.
\end{corollary}
\begin{proof}
Since \(p^{n-1}\vec{v}\) not within the span of a matrix implies that \(\vec{v}\) is of order \(p^n\), this satisfies the conditions of Theorem \ref{thm:2} with \(\mathfrak{B}\) being the two relevant subsystems and \(n'=n\), hence \(n\) copies of EPR pair(s) can be extracted.
\end{proof}
\section{Conclusion}
We generalize the decomposition of tripartite stabilizer states to arbitrary qudit dimensions, enabling entanglement manipulation for prime-power systems previously inaccessible with Clifford operations. This facilitates protocols such as entanglement shaping~\cite{zhong2021entanglement} and quantum secret sharing in arbitrary-level qudit systems. Future work could explore applying this method to more subsystems, the possibility of characterizing any \(m\)-partite stabilizer state with subsystem phase matrices and/or reducing all \(D=p^n\) stabilizer states to \(D=p\) stabilizer states with local unitaries beyond Clifford operations.
\section{Acknowledgments}
We acknowledge support from the ARO(W911NF-23-1-0077), ARO MURI (W911NF-21-1-0325), AFOSR MURI (FA9550-21-1-0209, FA9550-23-1-0338), DARPA (HR0011-24-9-0359, HR0011-24-9-0361), NSF (ERC-1941583, OMA-2137642, OSI-2326767, CCF-2312755, OSI-2426975), and the Packard Foundation (2020-71479). This material is based upon work supported by the U.S. Department of Energy, Office of Science, National Quantum Information Science Research Centers and Advanced Scientific Computing Research (ASCR) program under contract number DE-AC02-06CH11357 as part of the InterQnet quantum networking project.
\bibliography{savedrecs}
\end{document}